\documentclass{amsart}
\usepackage[title]{appendix}
\usepackage{amsmath}
\usepackage{amssymb,graphicx, enumerate,enumitem}
\usepackage{mathrsfs}
\usepackage{cmtt}
\usepackage{enumitem}

\usepackage[noline,ruled,titlenumbered,noend,norelsize]{algorithm2e}
\SetAlgoNoLine
\DontPrintSemicolon
\SetNlSty{footnotesize}{}{}
\IncMargin{0.5em}
\SetNlSkip{1em}
\SetKwIF{If}{ElseIf}{Else}{if}{then}{else if}{else}{endif}
\SetNlSkip{2em}
\newenvironment{algorithm-hbox}{\hbadness=10000\begin{algorithm}}{\end{algorithm}}

\newtheorem{theorem}{Theorem}
\newtheorem{claim}[theorem]{Claim}

\theoremstyle{remark}

\newtheorem{claim}[theorem]{Claim}

\let\leq\leqslant
\let\geq\geqslant

\let\subset\subseteq

\let\epsilon\varepsilon

\def\calG{\mathcal{G}}

\DeclareMathOperator{\color}{color}

\newcommand{\set}[1]{\left\{#1\right\}}

\renewenvironment{enumerate}{\begin{enumorig}[label=\textup{(\roman*)}, noitemsep, topsep=2pt plus 2pt, labelindent=.2em, leftmargin=*, widest=iii]}{\end{enumorig}}

\renewenvironment{itemize}{\begin{itemorig}[label=\textbullet, noitemsep, topsep=2pt plus 2pt, labelindent=.5em, labelsep=.5em, leftmargin=*]}{\end{itemorig}}



\begin{document}

\title[On-line algorithm for coloring bipartite graphs]{An on-line competitive algorithm for coloring bipartite graphs without long induced paths}

%

\author[P.~Micek]{Piotr Micek}
\address[P.~Micek]{Theoretical Computer Science Department, Faculty of Mathematics and Computer Science, Jagiellonian University, Krak\'{o}w, Poland}
\email{piotr.micek@tcs.uj.edu.pl}

\author[V.~Wiechert]{Veit Wiechert}
\address[V.~Wiechert]{Institut f\"ur Mathematik\\
 Technische Universit\"at Berlin\\
 Berlin \\
 Germany}
\email{wiechert@math.tu-berlin.de}

\begin{abstract}
The existence of an on-line competitive algorithm for coloring bipartite graphs remains a tantalizing open problem.
So far there are only partial positive results for bipartite graphs with certain small forbidden graphs as induced subgraphs. 
We propose a new on-line competitive coloring algorithm for $P_9$-free bipartite graphs. 
\end{abstract}

\thanks{This paper is an extended version of~\cite{MW-proc} from the proceedings of ISAAC 2014.}
\thanks{V.\ Wiechert is supported by Deutsche Forschungsgemeinschaft within the research training group `Methods for Discrete Structures' (GRK 1408).} 
\thanks{P.\ Micek is supported by Polish National Science Center UMO-2011/03/D/ST6/01370.}

\maketitle

\section{Introduction}
A \emph{proper coloring} of a graph is an assignment of colors to its vertices such that adjacent vertices receive distinct colors.
It is easy to devise an (linear time) algorithm for $2$-coloring bipartite graphs. 
Now, imagine that an algorithm receives vertices of a graph one by one knowing only the adjacency status of the vertex to vertices presented so far. 
The color of the current vertex must be fixed by the algorithm before the next vertex is revealed and it cannot be changed afterwards.
This kind of algorithm is called an \emph{on-line} coloring algorithm.

Formally, an \emph{on-line graph} $(G,\pi)$ is a graph $G$ with a permutation $\pi$ of its vertices.
An \emph{on-line coloring algorithm} $A$ takes an on-line graph $(G,\pi)$, say $\pi=(v_1,\ldots,v_n)$, as an input.
It produces a proper coloring of the vertices of $G$ where the color of a vertex $v_i$, for $i=1,\ldots,n$, depends only on the subgraph of $G$ induced by $v_1,\ldots,v_i$.
It is convenient to imagine that consecutive vertices along $\pi$ are revealed by some adaptive (malicious) adversary and the coloring process is a game between that adversary and the on-line algorithm.

Still, it is an easy exercise to show that if an adversary presents a bipartite graph and all the time the graph presented so far is connected then there is an on-line algorithm $2$-coloring these graphs.
But if an adversary can present a bipartite graph without any additional constraints then (s)he can trick out any on-line algorithm to use an arbitrary number of colors!

Indeed, there is a strategy for adversary forcing any on-line algorithm to use at least $\lfloor\log n\rfloor+1$ colors on a forest of size $n$. 
On the other hand, the First-Fit algorithm (that is an on-line algorithm coloring each incoming vertex with the least admissible natural number) uses at most $\lfloor\log n\rfloor+1$ colors on forests of size $n$.
When the game is played on bipartite graphs, an adversary can easily trick out First-Fit and force $\lceil\frac{n}{2}\rceil$ colors on a bipartite graph of size $n$.
Lov{\'a}sz, Saks and Trotter~\cite{LovaszST89} proposed a simple on-line algorithm (in fact as an exercise; see also \cite{Kierstead98}) using at most $2\log n+1$ colors on bipartite graphs of size $n$. 
This is best possible up to an additive constant as Gutowski et al.\ \cite{GutowskiKM14} showed that there is a strategy for adversary forcing any on-line algorithm to use at least $2\log n -10$ colors on a bipartite graph of size $n$.

For an on-line algorithm $A$ by $A(G,\pi)$ we mean the number of colors that $A$ uses against an adversary presenting graph $G$ with presentation order $\pi$.

An on-line coloring algorithm $A$ is \emph{competitive} on a class of graphs $\calG$ if there is a function $f$ such that for every $G\in\calG$ and permutation $\pi$ of vertices of $G$ we have $A(G,\pi)\leq f(\chi(G))$.
As we have discussed, there is no competitive coloring algorithm for forests. 
But there are reasonable classes of graphs admitting competitive algorithms, e.g., interval graphs can be colored on-line with at most $3\chi-2$ colors (where $\chi$ is the chromatic number of the presented graph; see~\cite{KiersteadT81}) and cocomparability graphs can be colored on-line with a number of colors bounded by a tower function in terms of $\chi$ (see~\cite{KiersteadPT94}).
Also classes of graphs defined in terms of forbidden induced subgraphs were investigated in this context.
For example, $P_4$-free graphs (also known as cographs) are colored by First-Fit optimally, i.e.\ with $\chi$ colors, since any maximal independent set meets all maximal cliques in a $P_4$-free graph. 
Also $P_5$-free graphs can be colored on-line with $O(4^\chi)$ colors (see \cite{KiersteadPT95}).
And to complete the picture there is no competitive algorithm for $P_6$-free graphs as Gy\'{a}rf\'{a}s and Lehel \cite{GL88} showed a strategy for adversary forcing any on-line algorithm to use an arbitrary number of colors on bipartite $P_6$-free graphs.

Confronted with so many negative results, it is not surprising that Gy\'{a}rf\'{a}s, Kir\'{a}ly and Lehel \cite{GKL97} introduced a relaxed version of competitiveness for on-line algorithms. 
The idea is to measure the efficiency of an on-line algorithm by comparing it to the best on-line algorithm for a given input (instead of the chromatic number).
Hence, the \emph{on-line chromatic number} of a graph $G$ is defined as
\[
 \chi_*(G) = \inf_A \max_{\pi} A(G,\pi),
\]
where the infimum is taken over all on-line algorithms $A$ and the maximum is taken over all permutation $\pi$ of vertices of $G$.
An on-line algorithm $A$ is \emph{on-line competitive} for a class of graphs $\calG$, if there is a function $f$ such that for every $G\in\calG$ and permutation $\pi$ of vertices of $G$ we have $A(G,\pi) \leq f(\chi_*(G))$.

Why are on-line competitive algorithms interesting?
Imagine that you design an algorithm and the input graph is not known in advance. 
If your algorithm is on-line competitive then you have an insurance that whenever your algorithm uses many colors on some graph $G$ with presentation order $\pi$ then any other on-line algorithm may be also forced to use many colors on the same graph $G$ with some presentation order $\pi'$ (and it includes also those on-line algorithms which are designed only for this single graph $G$!).
The idea of comparing the outputs of two on-line algorithms directly (not via the optimal off-line result) is present in the literature. 
We refer the reader to~\cite{BF07}, where a number of measures are discussed in the context of on-line bin packing problems.
In particular, the relative worst case ratio, introduced there, is closely related to our setting for on-line colorings.

It may be true that there is an on-line competitive algorithm for all graphs. 
This is open, as well as for the class of all bipartite graphs. 
To the best of the authors knowledge, there is no promissing approach for the negative answer for these questions.
However, there are some partial positive results. 
Gy\'{a}rf\'{a}s and Lehel \cite{GL90} have shown that First-Fit is on-line competitive for forests and it is even optimal in the sense that if First-Fit uses $k$ colors on $G$ then the on-line chromatic number of $G$ is $k$ as well.
They also have shown \cite{GKL97} that First-Fit is competitive (with an exponential bounding function) for graphs of girth at least $5$. 
Finally, Broersma, Capponi and Paulusma \cite{BCP08} proposed an on-line coloring algorithm for $P_7$-free bipartite graphs using at most $8\chi_*+8$ colors on graphs with on-line chromatic number $\chi_*$.

The contribution of this paper is the following theorem.

\begin{theorem}\label{thm:main}
There is an on-line competitive algorithm coloring $P_9$-free bipartite graphs
and using at most $6(\chi_*+1)^2$ colors, where $\chi_*$ is the on-line chromatic number of the presented graph.
\end{theorem}

Note that this is a full version of~\cite{MW-proc} published in the proceedings of ISAAC 2014.
In~\cite{MW-proc}, we discuss how our techniques simplify results for $P_7$-free bipartite graphs.
Already in~\cite{MW-proc}, we presented Algorithm~\ref{algo:bicolormax+} with a proof that it is on-line competitive for $P_8$-free bipartite graphs.
(This may be a good warmup or source of extra intuitions behind the argument in this paper.)

\section{Forcing structure}

In this section we introduce a family of bipartite graphs without long induced paths ($P_6$-free) and with arbitrarily large on-line chromatic number.
Our on-line algorithm, Algorithm~\ref{algo:bicolormax+} has the property that whenever it uses many colors on a graph $G$ 
then $G$ has a \emph{large} graph from our family as an induced subgraph and therefore $G$ has a large on-line chromatic number, as desired.

A connected bipartite graph $G$ has a unique partition of vertices into two independent sets. 
We call these partition sets the \emph{sides} of $G$.
A vertex $v$ in a bipartite graph $G$ is \emph{universal} to a subgraph $C$ of $G$ if $v$ is adjacent to all vertices of $C$ in one of the sides of $G$.

Consider a family of connected bipartite graphs $\set{X_k}_{k\geq1}$ defined recursively as follows. 
Each $X_k$ has a distinguished vertex called the \emph{root}. 
The side of $X_k$ containing the root of $X_k$, we call the \emph{root side} of $X_k$, while the other side we call the \emph{non-root side}.
$X_1$ is a single vertex being the root.
$X_2$ is a single edge with one of its vertices being the root.
$X_k$, for $k\geq3$, is a graph formed by two disjoint copies of $X_{k-1}$, say $X^1_{k-1}$ and $X^2_{k-1}$, with no edge between the copies, and one extra vertex $v$ adjacent to all vertices on the root side of $X^1_{k-1}$ and all vertices on the non-root side of $X^2_{k-1}$. 
The vertex $v$ is the root of $X_k$.
Note that for each $k$, the root of $X_k$ is adjacent to the whole non-root side of $X_k$, i.e., the root of $X_k$ is universal in $X_k$.
See Figure~\ref{fig:universal-struct} for an schematic illustration of the definition of $X_k$.

\begin{figure}[b]
 \centering
 \includegraphics{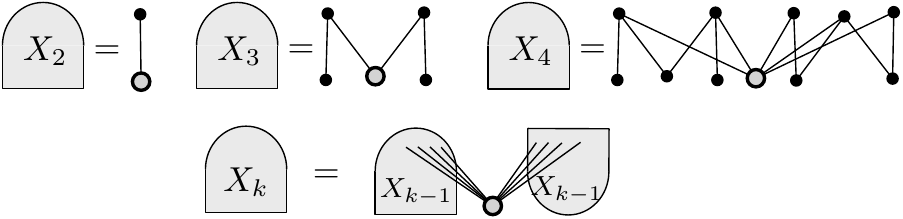}
 \caption{Family of bipartite graphs}
 \label{fig:universal-struct}
\end{figure}

A family of $P_6$-free bipartite graphs with arbitrarily large on-line chromatic number was first presented in \cite{GL88}.
The family $\set{X_k}_{k\geq1}$ was already studied in \cite{BCP06}, in particular Claim \ref{claim:forcing-structure} is proved there.
We encourage the reader to verify that $X_k$ is $P_6$-free for $k\geq1$.

\begin{claim}\label{claim:technical}
Let $k\geq 2$. 
Then, for every binary sequence $\alpha_1,\ldots,\alpha_{k-1}$, 
there are copies $Y_1,\ldots,Y_{k-1}$ of $X_1,\ldots,X_{k-1}$ contained as induced subgraphs in $X_k$, 
on pairwise disjoint sets of vertices and with no edges between the copies, such that 
for every $i\in\{1,\ldots,k-1\}$ the root side of $Y_i$ is contained in the root side of $X_k$, if $\alpha_i=1$, and
the non-root side of $Y_i$ is contained in the root side of $X_k$, if $\alpha_i=0$.
\begin{proof}
We prove the lemma by induction on $k$.
For the base case $k=2$, fix $\alpha_1\in\set{0,1}$. 
Now, $X_2$ is an edge and put $Y_1$ as a single vertex being the root of $X_2$ if $\alpha_1=1$ and being the other vertex if $\alpha_1=0$.
So suppose $k\geq3$.
By the definition $X_k$ consists of two independent copies $X_{k-1}^1$ and $X_{k-1}^2$ of $X_{k-1}$ and a root vertex that is universal to the root side of $X_{k-1}^1$ and to the non-root side of $X_{k-1}^2$.
Let us first consider the case $\alpha_{k-1}=0$.
Then $X_{k-1}^1$ is a copy of $X_{k-1}$ with the non-root on the desired side of $X_k$.
On $X_{k-1}^2$ we apply induction for the sequence $\alpha_1,\ldots,\alpha_{k-2}$ and find $Y_1,\ldots,Y_{k-2}$ copies of $X_1,\ldots,X_{k-2}$ as induced subgraphs of $X_{k-1}^2$ that are pairwise disjoint with no edges between the copies and clearly no edges to $X_{k-1}^1$ as well.
Since the root side of $X_{k-1}^2$ is contained in the root side of $X_k$ we have that for all $i\in\{1,\ldots,k-2\}$ the sides of the roots of $Y_i$ are in the root side of $X_k$ if and only if $\alpha_i=1$.

The case $\alpha_{k-1}=1$ is similar with the difference that we use $X_{k-1}^2$ as a copy of $X_{k-1}$ and that we apply induction on $X_{k-1}^1$ for the sequence $\overline{\alpha_1},\ldots,\overline{\alpha_{k-2}}$.
\end{proof}
\end{claim}

\begin{claim}\label{claim:forcing-structure}
 If $G$ contains $X_k$ as an induced subgraph, then $\chi_*(G)\geq k$.
\end{claim}

 \begin{proof}
  Let $A$ be any on-line coloring algorithm.
  We prove by induction on $k$ that the adversary can present the vertices of $G$ such that $A$ uses at least $k$ colors.
  It is clear that any coloring algorithm has to use one color for $X_1$ and two colors for $X_2$.
  So suppose that $k\geq3$.
  Adversary starts with presenting disjoint copies $Y_1,\ldots,Y_{k-1}$ of $X_1,\ldots,X_{k-1}$ one after another, with no edges between the copies, and by induction he can do this in such a way that $A$ uses $i$ colors on $Y_i$, for $i\in\set{1,\ldots,k-1}$.
  Therefore there are distinct colors $c_1,\ldots,c_{k-1}$ such that $c_i$ is used on $Y_i$ for every $i$
  and let $v_i\in Y_i$ be a vertex colored with $c_i$.
  Then we set $\alpha_i=1$ if $v_i$ is on the non-root side of $Y_i$ and $\alpha_i=0$ otherwise.
  
  Now we explain how to embed $Y_1,\ldots,Y_{k-1}$ into $G$.
  Let $v$ be the root of the induced copy of $X_k$ contained in $G$.
  By Claim \ref{claim:technical} there are pairwise disjoint induced copies of $X_1\ldots,X_{k-1}$ in $X_k$ with no edges between the copies, such that for all $i\in\{1,\ldots,k-1\}$ the root of $X_i$ is on the same side as $v$ if and only if $\alpha_i=1$. 
  Adversary identifies those copies with $Y_1,\ldots,Y_{k-1}$.
  By the choice of $\alpha_i$ it follows that $v_i$ is on the non-root side of $X_k$, for all $i\in\set{1,\ldots,k-1}$.
  Since $v$ is universal in $X_k$, it is adjacent to $v_i$ for all $i$.

  After presenting all $Y_1,\ldots,Y_{k-1}$ the adversary introduces vertex $v$, being the root of $X_k$ in $G$ and forces $A$ to use a color different from $c_1,\ldots,c_{k-1}$.
 \end{proof}

\section{The proof}

We present a new on-line algorithm for bipartite graphs, Algorithm~\ref{algo:bicolormax+}, and 
we prove that this algorithm is on-line competitive for $P_9$-free bipartite graphs.

Algorithm \ref{algo:bicolormax+} uses three disjoint pallettes of colors, $\set{a_n}_{n\geq1}$, $\set{b_n}_{n\geq1}$ and $\set{c_n}_{n\geq1}$.
In the following whenever the algorithm fixes a color of a vertex $v$ we are going to refer to it by $\color(v)$.
Also for any set of vertices $X$ we denote $\color(X)=\set{\color(x)\mid x\in X}$.
We say that $v$ has \emph{color index} $i$ if $\color(v)\in\set{a_i,b_i}$.

Suppose an adversary presents a new vertex $v$ of a bipartite graph $G$.
Then let $G_i[v]$ be the subgraph spanned by the vertices presented so far and colored with a color from $\set{a_1,\ldots,a_i,b_1,\ldots,b_i,c_1,\ldots,c_i}$ and vertex $v$, which is uncolored yet.
With $C_i[v]$ we denote the connected component of $G_i[v]$ containing $v$.
For convenience put $C_0[v]=\set{v}$.
Furthermore, let $C_i(v)$ be the graph $C_i[v]$ without vertex $v$.
For a vertex $x$ in $C_i(v)$ it will be convenient to denote by $C_i^x(v)$ the connected component of $C_i(v)$ that contains $x$.
We say that a color $c$ is \emph{mixed} in a connected subgraph $C$ of $G$ if $c$ is used on vertices on both sides of $C$.

\begin{algorithm-hbox}[ht]
    \caption{On-line competitive for $P_9$-free bipartite graphs} 
    \label{algo:bicolormax+}
    an adversary introduces a new vertex $v$\;
    $m \gets \max \set{i\geq1 \mid \textup{$a_i$ is mixed in $C_i[v]$}}+1$\tcp*[r]{$\max\set{}:=0$}
    \textup{let $I_1,I_2$ be the sides of $C_{m}[v]$ such that $v\in I_1$}\;
    \lIf{$a_{m}\in \color(I_2)$} {$\color(v)=b_{m}$}\;
    \lElse{\lIf{$c_{m}\in \color(I_2)$} {$\color(v)=a_{m}$}}\;
    \lElse{\lIf{$\exists u\in I_1\cup I_2$ and $\exists u'\in I_2$ such that $u$ has color index $j\geq m-\sqrt{2m}+2$ and $u'$ is universal to $C_{j-1}[u]$} {$\color(v)=c_{m}$}\;
    \lElse{$\color(v)=a_{m}$}}
\end{algorithm-hbox}

\begin{claim}\label{claim:proper-coloring-p9}
Algorithm \ref{algo:bicolormax+} gives a proper coloring of on-line bipartite graphs.
\end{claim}
\begin{proof}
  Suppose an adversary introduces a vertex $v$ of a bipartite graph $G$.
  We have to show that Algorithm~\ref{algo:bicolormax+} colors $v$ properly, i.e., no vertex presented before $v$ and adjacent to $v$ has the same color as $v$.
  Let $k\geq 1$ be the color index of $v$ and $(I_1,I_2)$ be the bipartition of $C_{k}[v]$ such that $v\in I_1$.
  If $\color(v)=a_k$, then there is no vertex in $I_2$ colored with $a_k$ because of the first if-condition.
  In particular, no neighbor of $v$ is colored with $a_k$.
  
  If $v$ is colored with $b_k$, then there is a vertex $u\in I_2$ with $\color(u)=a_k$.
  Suppose $v$ is not colored properly, which means that there is a vertex $w\in I_2$ with color $b_k$.
  When $w$ was introduced, there must have been a vertex $u'$ on the other side of $w$ in $C_{k}[w]$ with $\color(w)=a_k$.
  Since $C_k[w]\subseteq C_k[v]$ it follows that $u'\in I_1$.
  But then $u$ and $u'$ certify that color $a_k$ is mixed in $C_k[v]$, which contradicts the fact that the color index of $v$ is $k$.
  
  We are left with the case that $v$ is colored with $c_k$.
  Because of the second if-condition in the algorithm, a vertex can only get color $c_k$ if there is no vertex in $I_2$ colored in $c_k$, so in particular no neighbor of $v$ is colored with $c_k$.
 \end{proof}

The following claim, captures an idea behind maintaining the first two pallettes of colors (the $a_i$'s and $b_i$'s).
Namely, to force a single $a_i$-color to be mixed we need to introduce a vertex merging two components.
This idea is already present in previous works~\cite{BCP06,BCP08}.

\begin{claim}\label{claim:disjoint-components-p8}
Suppose an adversary presents a bipartite graph $G$ to Algorithm~\ref{algo:bicolormax+}.
Let $v\in V(G)$ and let $x,y$ be two vertices from opposite sides of $C_{i}[v]$ both colored with $a_i$ with $i\geq1$. 
Then $x$ and $y$ lie in different connected components of $C_{i}(v)$.
\end{claim}

\begin{proof}
Let $v$, $x$ and $y$ be like in the statement of the claim.
We are going to prove that at any moment after the introduction of $x$ and $y$, $x$ and $y$ lie in different connected components of the subgraph spanned by vertices colored with $a_1,b_1,c_1\ldots,a_i,b_i,c_i$.

Say $x$ is presented before $y$.
First note that $x\not\in C_i[y]$ as otherwise $x$ had to be on the opposite side to $y$ (because it is on the opposite side at the time $v$ is presented) and therefore $y$ would receive color $b_i$. 
Now consider any vertex $w$ presented after $y$ and suppose the statement is true before $w$ is introduced.
If $x\not\in C_i[w]$ or $y\not\in C_i[w]$ then whatever color is used for $w$ this vertex does not merge the components of $x$ and $y$ in the subgraph spanned by vertices presented so far and colored with $a_1,b_1,c_1\ldots,a_i,b_i,c_i$.
Otherwise $x,y\in C_i[w]$.
This means that color $a_{i}$ is mixed in $C_i[w]$ and therefore $w$ receives a color with an index at least $i+1$. 
Thus, the subgraph spanned by the vertices of the colors $a_1,b_1,c_1\ldots,a_i,b_i,c_i$ stays the same and $x$ and $y$ remain in different connected components of this graph.

Since all vertices in $C_i(v)$ are colored with $a_1,b_1,c_1\ldots,a_i,b_i,c_i$, we conclude that $x$ and $y$ lie in different components of $C_i(v)$.
\end{proof}

Consider a vertex $v$ with a color index $k\geq2$.
Let $v_1,v_2$ be the earliest introduced vertices from the opposite sides of $C_{k-1}(v)$, colored with $a_{k-1}$ (so witnessing that $a_{k-1}$ is mixed).
We call $v_1$ and $v_2$ the \emph{children} of $v$.
By Claim~\ref{claim:disjoint-components-p8} it follows that $C_{k-1}^{v_1}(v)$ and $C_{k-1}^{v_2}(v)$ are distinct (so disjoint and no edge is between them).

Suppose that an adversary presents a graph $G$ which is $P_9$-free.
Consider $v\in V(G)$ with color index $k\geq3$ and $v_1,v_2$ children of $v$.
Note that, at least one of $C_{k-1}^{v_1}[v]$ and $C_{k-1}^{v_2}[v]$ does not contain an induced $P_5$ with one endpoint in $v$.
Indeed, the join of two such paths at $v$ would end up in an induced $P_9$, which is forbidden in $G$.
Choose arbitrarily a component $C_{k-1}^{v_i}[v]$ with no induced $P_5$ ending at $v$ and let $v_{i,1},v_{i,2}$ be the children of $v_i$.
We call $v_{i,1}$ and $v_{i,2}$ the \emph{grandchildren} of $v$.

The next claim describes a property of a component containing grandchildren of a given vertex, namely, under certain condition we win a universal vertex to a subcomponent.
The usage of the third pallette of colors, the $c_i$'s, is inspired by this property.

 \begin{claim}\label{claim:path-vs-universal}  
  Suppose an adversary presents a $P_9$-free bipartite graph $G$ to Algorithm~\ref{algo:bicolormax+}.
  Let $x$ be a vertex with color index $i\geq 2$.
  Suppose that vertex $y\in C_{i-1}(x)$, with color index $j$, lies on the other side of $x$ in $G$ and $y$ is not adjacent to $x$.
  If there is no induced path of length 5 in $C^y_{i-1}[x]$ with one endpoint in $x$, then $x$ has a neighbor in $C^y_{i-1}(x)$ that is universal to $C_{j-1}[y]$.  
 \end{claim}
 \begin{proof}
We can assume that $y$ has color index $j\geq2$, as otherwise $C_{j-1}[y]=C_0[y]=\set{y}$ and vacuously any neighbor of $x$ is universal to $C_{j-1}[y]$ (as the side it should be adjacent to is empty).

First, let us consider the case that $x$ has a neighbor $z$ in $C_{j-1}[y]\subseteq C_{i-1}^y(x)$ (see Figure~\ref{fig:one-side-universal} for this case).
Since $x$ and $y$ are not adjacent we have $y\neq z$.
As $C^z_{j-1}[y]$ is connected, there is an induced path $P$ connecting $x$ and $y$ that has only vertices of $C^z_{j-1}(y)$ as inner vertices.
Clearly, $P$ has even length at least $4$.
Now the color index of $y$, namely $j\geq2$, assures the existence of a mixed pair coloered with $a_{j-1}$ in $C_{j-1}[y]$ and with Claim~\ref{claim:disjoint-components-p8} it follows that $C_{j-1}(y)$ has at least two connected components.
In particular, there is a component $C'$ of $C_{j-1}(y)$ other than $C^z_{j-1}(y)$.
Clearly, $y$ has a neighbor $z'$ in $C'$, which we use to prolong $P$ at $y$.
Since there is no edge between $C^z_{j-1}(y)$ and $C'$, vertex $z'$ is not adjacent to the inner vertices of $P$.
And as $G$ is bipartite $z'$ cannot be adjacent to $x$.
We conclude the existence of an induced path of length $5$ in $C^y_{i-1}[x]$ with $x$ and $z'$ being its endpoints, a contradiction.
   
 \begin{figure}[t]
 \centering
 \includegraphics{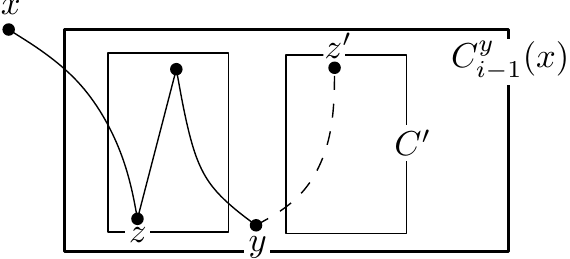}
 \caption{Claim~\ref{claim:path-vs-universal}: Situation in which $x$ has a neighbor $z$ in $C_{j-1}[y]$.}
 \label{fig:one-side-universal}
\end{figure}
   
   Second, we consider the case that $x$ has no neighbor in $C_{j-1}[y]$.
   By our assumptions a shortest path connecting $x$ and $y$ in $C_{i-1}^y[x]$ must have length exactly $4$.
   Let $P=(x,r,s,y)$ be such a path.   
   We claim that vertex $r$ is universal to $C_{j-1}[y]$.
   Suppose to contrary that there is a vertex $s'$ in $C_{j-1}[y]$ which is on the other side of $y$ and which is not adjacent to $r$.
   Let $Q=(y,s_1,r_1,\ldots,s_{\ell-1},r_{\ell-1},s_{\ell}=s')$ be a shortest path connecting $y$ and $s'$ in $C_{j-1}[y]$.
   For convenience put $s_0=s$.
   Now we choose the minimal $m\geq 0$ such that $r$ is adjacent to $s_m$ but not to $s_{m+1}$.
   Such an $m$ exists since $r$ is adjacent to $s_0=s$ but not to $s_{\ell}$.
   If $m=0$ then the path $(x,r,s,y,s_1)$ is an induced path of length $5$ and if $m>0$ then the path $(x,r,s_m,r_m,s_{m+1})$ has length $5$ and is induced unless $x$ and $r_m$ are adjacent.
   But the latter is not possible since $x$ has no neighbor in $C_{j-1}[y]$.
   Thus, in both cases we get a contradiction and we conclude that $r$ is universal to $C_{j-1}[y]$.
  \end{proof}

In the following we write $v\rightarrow_i w$ for $v,w\in V(G)$, if there is a sequence $v=x_1,\ldots,x_j=w$ with $j\leq i$ and $x_{\ell+1}$ is a grandchild of $x_{\ell}$, for all $\ell\in\set{1,\ldots,j-1}$. 
Moreover, we define $S_i(v)=\set{w\mid v\rightarrow_i w}$. 
Thus, $S_1(v)=\set{v}$ for every $v\in V(G)$.

We make some immediate observations concerning this definition.
Let $v\in V(G)$ be a vertex with color index $k\geq 3$.
Then each $w\in S_i(v)$ has color index at least $k-2i+2$ and the component $C_{k-1}^w(v)$ does not contain an induced $P_5$ with one endpoint in $v$.
Furthermore, each vertex in $S_i(v)$ is connected to $v$ by a path in $G$ and all vertices in the path, except $v$, have color index at most $k-1$.
This proves that $S_i(v) \subset C_{k-1}[v]$, for all $i\geq1$.
Note also that if $v_1$ and $v_2$ are the grandchildren of $v$ then we have $S_i(v)=\set{v}\cup S_{i-1}(v_1) \cup S_{i-1}(v_2)$.
By definition $v_1$ and $v_2$ are the children of a child $v'$ of $v$.
It follows that $S_{i-1}(v_1)\subset C_{k-2}^{v_1}(v')$ and $S_{i-1}(v_2)\subset C_{k-2}^{v_2}(v')$.
By Claim~\ref{claim:disjoint-components-p8}, we get that $C_{k-2}^{v_1}(v')$ and $C_{k-2}^{v_2}(v')$ are distinct.
In particular, $S_{i-1}(v_1)$ and $S_{i-1}(v_2)$ are disjoint and there is no edge between them.

For a vertex $v\in V(G)$, $S_{i}(v)$ is \emph{complete} in $G$ if for every $u,w\in S_i(v)$ such that $u\rightarrow_i w$ and $u,w$ lying on opposite sides of $G$, we have $u$ and $w$ being adjacent in $G$.
Note that $v$ is a universal vertex in $S_i(v)$, provided $S_i(v)$ is complete.

 \begin{claim}\label{claim:full-tree}
Suppose an adversary presents a bipartite graph $G$ to Algorithm~\ref{algo:bicolormax+}.
Let $v\in V(G)$ be a vertex with color index $k$ and let $k\geq 2i \geq2$.
If $S_{i}(v)$ is complete then $S_{i}(v)$ contains an induced copy of $X_{i}$ in $G$ with $v$ being the root of the copy.
\end{claim}
 \begin{proof}
 We prove the claim by induction on $i$.
 For $i=1$ we work with $S_1(v)$ and $X_1$ being graphs with one vertex only, so the statement is trivial.
 For $i\geq2$, let $v_1$ and $v_2$ be the grandchildren of $v$.
 Recall that $S_i(v)=\{v\}\cup S_{i-1}(v_1)\cup S_{i-1}(v_2)$.
 Since $S_i(v)$ is complete it also follows that $S_{i-1}(v_1)$ and $S_{i-1}(v_2)$ are complete.
 So by the induction hypothesis there are induced copies $X_{i-1}^1$, $X_{i-1}^2$ of $X_{i-1}$ in $S_{i-1}(v_1)$ and $S_{i-1}(v_2)$, respectively, and rooted in $v_1$, $v_2$, respectively.
 Recall that $S_{i-1}(v_1)$ and $S_{i-1}(v_2)$ are disjoint and there is no edge between them.
 Thus, the copies $X_{i-1}^1$ and $X_{i-1}^2$ of $X_{i-1}$ are disjoint and there is no edge between them, as well.
 Since $S_i(v)$ is complete $v$ is universal to both of the copies, and since $v_1$ and $v_2$ lie on opposite sides in $G$ we get that the vertices of $X_{i-1}^1 \cup X_{i-1}^2 \cup \set{v}$ induce a copy of $X_i$ in $G$.
 \end{proof}
 
 \begin{claim}\label{claim:P8-main-claim}
  Suppose an adversary presents a $P_9$-free bipartite graph $G$ to Algorithm~\ref{algo:bicolormax+} and suppose vertex $v$ is colored with $a_{k}$ and $k\geq2$.
  Then $C_{k}[v]$ contains an induced copy of $X_{\lfloor\sqrt{k/2}\rfloor}$ such that its root lies on the same side as $v$ in $G$.
 \end{claim}
  \begin{proof}
   We prove the claim by induction on $k$.
   For $k=2$ the statement is trivial.
   So suppose that $k\geq 3$.
   If $S_{\lfloor\sqrt{k/2}\rfloor}(v)$ is complete then by Claim~\ref{claim:full-tree} we get an induced copy of $X_{\lfloor\sqrt{k/2}\rfloor}$ with a root mapped to $v$, as required.
   
   From now on we assume that $S_{\lfloor\sqrt{k/2}\rfloor}(v)$ is not complete.
   Let $(I_1,I_2)$ be the bipartition of $C_{k}[v]$ such that $v\in I_1$.
   First, we will prove that there are vertices $z,z'\in C_{k}[v]$ such that $z'\in I_1$, $z$ has color index $\ell \geq k-\sqrt{2k}+2$ and $z'$ is universal to $C_{\ell-1}[z]$.
   To do so we consider the reason why Algorithm~\ref{algo:bicolormax+} colors $v$ with $a_{k}$.
   
   The first possibility is that the second if-condition of the algorithm is satisfied, that is, there is a vertex $u\in I_2$ colored with $c_{k}$.
   Now $u$ can only receive color $c_{k}$ if there are vertices $w, w'\in C_{k}[u]$ such that $w'$ is on the other side of $u$ in $C_{k}[u]$, $w$ has color index $j\geq k-\sqrt{2k}+2$ 
   and $w'$ is universal to $C_{j-1}[w]$.
   Since $C_{k}(u)\subseteq C_{k}(v)$ and $u\in I_2$ we have $w'\in I_1$.
   Therefore, $z=w$ and $z'=w'$ are vertices we are looking for.
   
   The second reason for coloring $v$ with $a_{k}$ could be that Algorithm~\ref{algo:bicolormax+} reaches its last line.
   In particular this means, that there is no vertex of color $a_{k}$ or $c_{k}$ in $I_2$.
   Now we are going to make use of the fact that $S_{\lfloor\sqrt{k/2}\rfloor}(v)$ is not complete.
   There are vertices $x,y\in S_{\lfloor\sqrt{k/2}\rfloor}(v)\subseteq C_{k}[v]$ such that $x\rightarrow_{\lfloor\sqrt{k/2}\rfloor} y$, the vertices $x$ and $y$ lie on different sides of $C_{k}[v]$ and are not adjacent.
   Let $i$ and $j$ be the color indices of $x$ and $y$, respectively.
   Note that $k\geq i > j\geq k-2\lfloor\sqrt{k/2}\rfloor+2$.
   By the definition of a grandchild it follows that $C_{i-1}^y[x]$ does not contain an induced $P_5$ with one endpoint in $x$.
   Hence we can apply Claim~\ref{claim:path-vs-universal} and it follows that $x$ has a neighbor $r\in C_{i-1}^y(x)$ that is universal to $C_{j-1}[y]$.
   We set $z'=r$ and $z=y$.
   Then, we have that $z'$ is universal to $C_{j-1}[z]$ with
   \[
    j \geq k-2\lfloor\sqrt{k/2}\rfloor+2\geq k-\sqrt{2k}+2.
   \]
   Since $z'\in C_k[v]$ we have $z'\in I_1$ or $z'\in I_2$.   
   However, the latter is not possible as otherwise $z$ and $z'$ would fulfill the conditions of the third if-statement in Algorithm~\ref{algo:bicolormax+}, which contradicts the fact that Algorithm~\ref{algo:bicolormax+} reached the last line while processing $v$.
   We conclude that $z'\in I_1$, which completes the proof of our subclaim.
   
   Now fix $z,z'$ witnessing our subclaim.
   Let $z_1$ and $z_2$ be the children of $z$.
   Both vertices received color $a_{\ell-1}$ and they are on different sides of $G$.
   By the induction hypothesis $C_{\ell-1}[z_1]$ and $C_{\ell-1}[z_2]$ contain a copy of $X_{\lfloor\sqrt{(\ell-1)/2}\rfloor}$ such that the roots are on the same side as $z_1$ and $z_2$, respectively.
   Since there is no edge between $C_{\ell-1}[z_1]$ and $C_{\ell-1}[z_2]$ (this is a consequence of Claim~\ref{claim:disjoint-components-p8}) and both are contained in $C_{\ell-1}[z]$, it follows that $z'$ together with the copies of $X_{\lfloor\sqrt{(\ell-1)/2}\rfloor}$ induce a copy of $X_{\lfloor\sqrt{(\ell-1)/2}\rfloor+1}$ that has $z'$ as its root (see Figure~\ref{fig:claim8}).
   Since $C_{\ell-1}[z]$ is contained in $C_{k}[v]$ and $z'$ is on the same side as $v$ and since
   \[
   \lfloor\sqrt{(\ell-1)/2}\rfloor+1\geq
   \left\lfloor\sqrt{(k-\sqrt{2k}+1)/2}\right\rfloor+1\geq \left\lfloor\sqrt{k/2-\sqrt{k/2}}\right\rfloor +1\geq \left\lfloor\sqrt{k/2}\right\rfloor,
   \]
   for all $k\geq 0$, the proof is complete.
   \end{proof}
   
   \begin{figure}[t]
    \centering
    \includegraphics{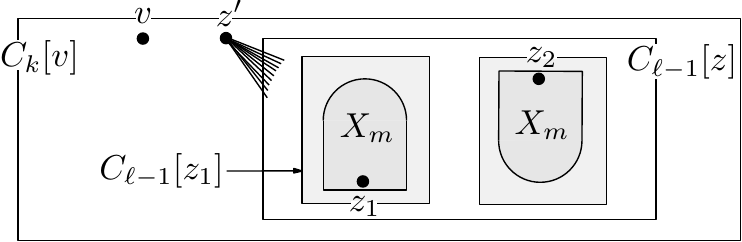}
    \caption{Final step in Claim~\ref{claim:P8-main-claim}. The value $m$ stands for $\lfloor\sqrt{(\ell-1)/2}\rfloor$.}
    \label{fig:claim8}
   \end{figure}

Now we are able to prove our main theorem.
\begin{proof}[Proof of Theorem~\ref{thm:main}]
  Let $k$ be the largest color index used by Algorithm~\ref{algo:bicolormax+} while coloring vertices of $G$.
  In particular, Algorithm~\ref{algo:bicolormax+} uses at most $3k$ colors for $G$.
  If $k=1$ then the statement is obvious.
  Suppose $k\geq2$.
  There must be a vertex in $G$ colored with $a_k$.
  By Claim~\ref{claim:P8-main-claim} it follows that $G$ contains $X_{\lfloor\sqrt{k/2}\rfloor}$
  and by Claim~\ref{claim:forcing-structure}, $\chi_*(G)\geq \lfloor\sqrt{k/2}\rfloor \geq \sqrt{k/2} -1$.
  This together with $3k= 6(\sqrt{k/2}-1+1)^2\leq 6(\chi_*(G)+1)^2$ completes the proof.
 \end{proof}

\bibliographystyle{plain}
\bibliography{online-comp}

\begin{thebibliography}{10}

\bibitem{BF07}
Joan Boyar and Lene~M. Favrholdt.
\newblock The relative worst order ratio for online algorithms.
\newblock {\em ACM Trans. Algorithms}, 3(2):Art. 22, 24, 2007.

\bibitem{BCP06}
Hajo~J. Broersma, Agostino Capponi, and Dani\"{e}l Paulusma.
\newblock On-line coloring of {$H$}-free bipartite graphs.
\newblock In Tiziana Calamoneri, Irene Finocchi, and Giuseppe~F. Italiano,
  editors, {\em Algorithms and Complexity}, volume 3998 of {\em Lecture Notes
  in Computer Science}, pages 284--295. Springer Berlin Heidelberg, 2006.

\bibitem{BCP08}
Hajo~J. Broersma, Agostino Capponi, and Dani\"{e}l Paulusma.
\newblock A new algorithm for on-line coloring bipartite graphs.
\newblock {\em SIAM Journal on Discrete Mathematics}, 22(1):72--91, February
  2008.

\bibitem{GutowskiKM14}
Grzegorz Gutowski, Jakub Kozik, and Piotr Micek.
\newblock Lower bounds for {O}n-line {G}raph {C}olorings.
\newblock In {\em 25th International Symposium, ISAAC 2014, Jeonju, Korea,
  December 15-17, 2014, Proceedings}, volume 8889 of {\em Lecture Notes in
  Computer Science}, pages 507--515, 2014.

\bibitem{GKL97}
Andr\'{a}s Gy\'{a}rf\'{a}s, Zolt\'{a}n Kir\'{a}ly, and Jen{\H{o}} Lehel.
\newblock On-line competitive coloring algorithms.
\newblock Technical report TR-9703-1, available online at
  http://www.cs.elte.hu/tr97/tr9703-1.ps, 1997.

\bibitem{GL88}
Andr\'{a}s Gy\'{a}rf\'{a}s and Jen\H{o} Lehel.
\newblock On-line and first fit colorings of graphs.
\newblock {\em Journal of Graph Theory}, 12(2):217--227, 1988.

\bibitem{GL90}
Andr{\'a}s Gy{\'a}rf{\'a}s and Jen{\H{o}} Lehel.
\newblock First fit and on-line chromatic number of families of graphs.
\newblock {\em Ars Combin.}, 29(C):168--176, 1990.
\newblock Twelfth British Combinatorial Conference (Norwich, 1989).

\bibitem{Kierstead98}
Henry~A. Kierstead.
\newblock Recursive and on-line graph coloring.
\newblock In {\em Handbook of recursive mathematics, {V}ol.\ 2}, volume 139 of
  {\em Stud. Logic Found. Math.}, pages 1233--1269. North-Holland, Amsterdam,
  1998.

\bibitem{KiersteadPT94}
Henry~A. Kierstead, Stephen~G. Penrice, and William~T. Trotter.
\newblock On-line coloring and recursive graph theory.
\newblock {\em SIAM Journal on Discrete Mathematics}, 7(1):72--89, 1994.

\bibitem{KiersteadPT95}
Henry~A. Kierstead, Stephen~G. Penrice, and William~T. Trotter.
\newblock On-line and first-fit coloring of graphs that do not induce {$P_5$}.
\newblock {\em SIAM Journal on Discrete Mathematics}, 8(4):485--498, 1995.

\bibitem{KiersteadT81}
Henry~A. Kierstead and William~T. Trotter.
\newblock An extremal problem in recursive combinatorics.
\newblock In {\em Proceedings of the {T}welfth {S}outheastern {C}onference on
  {C}ombinatorics, {G}raph {T}heory and {C}omputing, {V}ol. {II}}, volume~33 of
  {\em Congressus Numerantium}, pages 143--153, 1981.

\bibitem{LovaszST89}
L{\'a}szl{\'o} Lov{\'a}sz, Michael~E. Saks, and William~T. Trotter.
\newblock An on-line graph coloring algorithm with sublinear performance ratio.
\newblock {\em Discrete Mathematics}, 75(1-3):319--325, 1989.

\bibitem{MW-proc}
Piotr Micek and Veit Wiechert.
\newblock An {O}n-line {C}ompetitive {A}lgorithm for {C}oloring {$P_8$}-free
  {B}ipartite {G}raphs.
\newblock In {\em 25th International Symposium, ISAAC 2014, Jeonju, Korea,
  December 15-17, 2014, Proceedings}, volume 8889 of {\em Lecture Notes in
  Computer Science}, pages 516--527, 2014.

\end{thebibliography}

\end{document}